\documentclass[11pt, draftcls, onecolumn]{IEEEtran}
\usepackage{subfigure}
\usepackage{setspace}
\usepackage{amsmath}
\usepackage{amssymb}
\usepackage{amsfonts}
\usepackage{amscd}
\usepackage{mathrsfs}
\usepackage[final]{graphicx}
\usepackage{graphicx}
\usepackage{psfrag}
\usepackage{epsfig}
\usepackage{color}
\usepackage{url}
\usepackage{textcomp}
\usepackage{multirow}
\input{epsf.sty}
\newtheorem{theorem}{Theorem}
\newtheorem{lemma}{Lemma}
\newtheorem{definition}{Definition}

\begin{document}

\title{Reduced ML-Decoding Complexity, Full-Rate STBCs for $2^a$ Transmit Antenna Systems}
\author{
\authorblockN{K. Pavan Srinath and B. Sundar Rajan,\\}
\authorblockA{Dept of ECE, Indian Institute of science, \\
Bangalore 560012, India\\
Email:\{pavan,bsrajan\}@ece.iisc.ernet.in\\
}
}
\maketitle
\begin{abstract}
For an $n_t$ transmit, $n_r$ receive antenna system ($n_t \times n_r$ system), a {\it{full-rate}} space time block code (STBC) transmits $n_{min} = min(n_t,n_r)$ complex symbols per channel use and in general, has an ML-decoding complexity of the order of $M^{n_tn_{min}}$ (considering square designs), where $M$ is the constellation size. In this paper, a scheme to obtain a full-rate STBC for $2^a$ transmit antennas and any $n_r$, with reduced ML-decoding complexity of the order of $M^{n_t(n_{min}-\frac{3}{4})}$, is presented. The weight matrices of the proposed STBC are obtained from the unitary matrix representations of a Clifford Algebra. For any value of $n_r$, the proposed design offers a reduction from the full ML-decoding complexity by a factor of $M^{\frac{3n_t}{4}}$. The well known Silver code for 2 transmit antennas is a special case of the proposed scheme. Further, it is shown that the codes constructed using the scheme have higher ergodic capacity than the well known punctured Perfect codes for $n_r < n_t$. Simulation results of the symbol error rates are shown for $8 \times 2$ systems, where the comparison of the proposed code is with the punctured Perfect code for 8 transmit antennas. The proposed code matches the punctured perfect code in error performance, while having reduced ML-decoding complexity and higher ergodic capacity.
\end{abstract}

\section{Introduction and Background}
Complex orthogonal designs (CODs) \cite{TJC}, \cite{TiH}, although they provide linear Maximum Likelihood (ML) decoding, do not offer a high rate of transmission. A full-rate code for an $n_t \times n_r$ MIMO system transmits $min(n_t,n_r)$ complex symbols per channel use. Among the CODs, only the Alamouti code for 2 transmit antennas is full-rate for a $2 \times 1$ MIMO system. A full-rate STBC can efficiently utilize all the degrees of freedom the channel provides. An increase in the rate also results in an increase in the ML-decoding complexity. The Golden code \cite{BRV} for 2 transmit antennas is an example of a full-rate STBC for any number of receive antennas. Until recently, the ML-decoding complexity of the Golden code was reported to be of the order of $M^4$, where $M$ is the size of the signal constellation. However, it was shown in \cite{john_barry1}, \cite{SrR_arxiv} that the Golden code has a decoding complexity of the order of $M^2\sqrt{M}$ only. A lot of attention is being given to reducing the ML-complexity of full-rate codes. Current research focuses on obtaining high rate codes with reduced ML-decoding complexity (refer to Sec. \ref{sec2} for a formal definition), since high rate codes are essential to exploit the available degrees of freedom of the MIMO channel. For 2 transmit antennas, the Silver code \cite{HTW}, \cite{PGA}, is a full-rate code with full-diversity and an ML-decoding complexity of order $M^2$ for square QAM. For 4 transmit antennas, Biglieri et. al. proposed a rate-2 STBC which has an ML-decoding complexity of $M^4\sqrt{M}$ for square QAM without full-diversity \cite{BHV}. It was, however, shown that there was no significant reduction in error performance at low to medium SNR when compared with the previously best known code - the DjABBA code \cite{HTW}. This code was obtained by multiplexing Quasi-orthogonal designs(QOD) for 4 transmit antennas \cite{JH}. In \cite{SrR_arxiv}, a new full-rate STBC for $4\times 2$ system with full diversity and an ML-decoding complexity of $M^4\sqrt{M}$ was proposed. This code was obtained by multiplexing the coordinate interleaved orthogonal designs (CIODs) for 4 transmit antennas \cite{ZS}. These results show that codes obtained by multiplexing low complexity STBCs can result in high rate STBCs with reduced ML-decoding complexity and without any significant degradation in the error performance when compared with the best existing STBCs. Such an approach has also been adopted in \cite{Robert} to obtain high rate codes from multiplexed orthogonal designs.

In general, it is not known how one can design full-rate STBCs for arbitrary number of transmit and receive antennas with reduced ML-decoding complexity. Such a design has been presented for $n_t = 4$ in \cite{srinath}. It is known how to design information lossless codes \cite{damen_info} for the case where $n_r \geq n_t$. However, it is not known how to design information lossless codes when $n_r < n_t$. In this paper, we design codes which have higher ergodic capacity at high signal to noise ratio ($SNR$) than the best existing codes (the Perfect codes \cite{new_per}) for $n_r < n_t$. The resulting codes also have lower ML-decoding complexity than the comparable punctured Perfect codes. The contributions of the paper are: 

\begin{enumerate}
\item We analyze the ergodic capacity of MIMO channels with space time codes when $n_r < n_t$. We relate the entries of the $\textbf{R}$-matrix of the equivalent channel matrix to ergodic capacity at high $SNR$. 
\item We give a scheme to obtain rate-1, 4-group decodable codes (refer Section \ref{sec2} for a formal definition of multigroup decodable codes) for $n_t = 2^a$ through algebraic methods. The speciality of the obtained design is that it is amenable for extension to higher number of receive antennas, resulting in full-rate, reduced ML-decoding complexity codes for any number of receive antennas, unlike the previous constructions \cite{4gp1}, \cite{4gp2}, \cite{sanjay} of rate-1, 4-group decodable codes.
\item We propose a scheme to obtain full-rate, reduced ML-decoding complexity codes for $2^a$ transmit antennas and any number of receive antennas. These codes are also shown to have higher ergodic capacity than the comparable punctured Perfect codes for the case $n_r < n_t$, and lower ML-decoding complexity as well. In terms of error performance, the proposed codes have more or less the same performance as the corresponding punctured Perfect codes. This is shown through simulation results for the $8 \times 2$ MIMO system.
\end{enumerate}

The paper is organized as follows. In Section \ref{sec2}, we present the system model and the relevant definitions. The ergodic capacity analysis is presented in Section \ref{sec3} and the method to construct Rate-1, 4-group decodable codes is proposed in Section \ref{sec4}. The scheme to extend the code to obtain full-rate STBCs for higher number of receive antennas is presented in Section \ref{sec5}. Simulation results are discussed in Section \ref{sec6} and the concluding remarks are made in Section \ref{sec7}.

\textit{\textbf{Notations}:} Throughout, bold, lowercase letters are used to denote vectors and bold, uppercase letters are used to denote matrices. Let $\textbf{X}$ be a complex matrix. Then, $\textbf{X}^{H}$ and $\textbf{X}^{T}$ denote the Hermitian and the transpose of $\textbf{X}$, respectively and $j$ represents $\sqrt{-1}$. The $(i,j)^{th}$ entry of $\textbf{X}$ is denoted by $\textbf{X}(i,j)$ and $tr(\textbf{X})$ denotes the trace of $\textbf{X}$. The set of all real and complex numbers are denoted by $\mathbb{R}$ and $\mathbb{C}$, respectively. The real and the imaginary part of a complex number $x$ are denoted by $x_I$ and $x_Q$, respectively. $\Vert \textbf{X} \Vert$ denotes the Frobenius norm of $\textbf{X}$, and $\textbf{I}_T$ and $\textbf{O}_T$ denote the $T\times T$ identity matrix and the null matrix, respectively. The Kronecker product is denoted by $\otimes$. For a complex random variable $X$, $\mathcal{E}[X]$ denotes the mean of $X$. The inner product of two vectors $\textbf{x}$ and $\textbf{y}$ is denoted by $\langle \textbf{x},\textbf{y} \rangle$.

For a complex variable $x$, the $\check{(.)}$ operator acting on $x$ is defined as 
\begin{equation*}
\check{x} \triangleq \left[ \begin{array}{rr}
                             x_I & -x_Q \\
                             x_Q & x_I \\
                            \end{array}\right].
\end{equation*}
The $\check{(.)}$ can similarly be applied to any matrix $\textbf{X} \in \mathbb{C}^{n \times m}$ by replacing each entry $x_{ij}$ by $\check{x}_{ij}$, $i=1,2,\cdots, n, j = 1,2,\cdots,m$ , resulting in a matrix denoted by $\check{\textbf{X}} \in \mathbb{R}^{2n \times 2m}$.

Given a complex vector $\textbf{x} = [ x_1, x_2, \cdots, x_n ]^T$, $\tilde{\textbf{x}}$ is defined as
\begin{equation*}
\tilde{\textbf{x}} \triangleq [ x_{1I},x_{1Q}, \cdots, x_{nI}, x_{nQ} ]^T.
\end{equation*}

\section{System Model}
\label{sec2}
We consider Rayleigh block fading MIMO channel with full channel state information (CSI) at the receiver but not at the transmitter. For $n_t \times n_r$ MIMO transmission, we have
\begin{equation}\label{Y}
\textbf{Y} = \sqrt{\frac{SNR}{n_t}}\textbf{HS + N},
\end{equation}

\noindent where $\textbf{S} \in \mathbb{C}^{n_t \times T}$ is the codeword matrix whose average energy is given by $\mathcal{E}(\Vert \textbf{S} \Vert^2) = n_tT$, transmitted over $T$ channel uses, $\textbf{N} \in \mathbb{C}^{n_r \times T}$ is a complex white Gaussian noise matrix with i.i.d entries $\sim
\mathcal{N}_{\mathbb{C}}\left(0,1\right)$ and $\textbf{H} \in \mathbb{C}^{n_r\times n_t}$ is the channel matrix with the entries assumed to be i.i.d circularly symmetric Gaussian random variables $\sim \mathcal{N}_\mathbb{C}\left(0,1\right)$. $\textbf{Y} \in \mathbb{C}^{n_r \times T}$ is the received matrix and $SNR$ is the signal to noise ratio at each receive antenna.

\begin{definition}\label{def1}$\left(\textbf{Code rate}\right)$ Code rate is the average number of independent information symbols transmitted per channel use. If there are $k$ independent complex information symbols (or $2k$ real information symbols) in the codeword which are transmitted over $T$ channel uses, then, the code rate is $k/T$ complex symbols per channel use ($2k/T$ real symbols per channel use).
\end{definition}

\begin{definition}\label{def2}$\left(\textbf{Full-rate STBCs}\right)$ For an $n_t \times n_r$ MIMO system, if the code rate is $min\left(n_t,n_r\right)$ complex symbols per channel use, then the STBC is said to be \emph{\textbf{full-rate}}.
\end{definition}

 Assuming ML-decoding, the ML-decoding metric that is to be minimized over all possible values of codewords $\textbf{S}$ is given by
 \begin{equation}
\label{ML}
 \textbf{M}\left(\textbf{S}\right) = \Vert \textbf{Y} - \sqrt{\frac{SNR}{n_t}}{}\textbf{HS} \Vert^2
 \end{equation}

\begin{definition}\label{def3}$\left(\textbf{ML-Decoding complexity}\right)$
The ML decoding complexity is measured in terms of the maximum number of symbols that need to be jointly decoded in minimizing the ML decoding metric.
\end{definition}
For example, if the codeword transmits $k$ independent symbols of which a maximum of $p$ symbols need to be jointly decoded, the ML-decoding complexity is of the order of $M^{p}$, where $M$ is the size of the signal constellation. If the code has an ML-decoding complexity of order less than $M^k$, the code is said to admit \emph{\textbf{reduced ML-decoding}}.

\begin{definition}\label{def4}$\left(\textbf{Generator matrix}\right)$ For any STBC that encodes $2k$ real symbols (or $k$ complex information symbols), the \textbf{\emph{generator}} matrix $\textbf{G}$ is defined by the following equation  \cite{BHV}.
\begin{equation*}
\widetilde{vec\left(\textbf{S}\right)} = \textbf{G} \textbf{s},
\end{equation*}
\noindent where $\textbf{S}$ is the codeword matrix, $\textbf{s} \triangleq \left[ s_1, s_2,\cdots,s_{2k} \right]^T$ is the real information symbol vector.
\end{definition}

A codeword matrix of an STBC can be expressed in terms of \textbf{\emph{weight matrices}} (linear dispersion matrices) \cite{HaH} as 
\begin{equation*}
\textbf{S} = \sum_{i=1}^{2k}s_{i}\textbf{A}_{i} .
\end{equation*}
Here, $\textbf{A}_i,i=1,2,\cdots,2k$ are the complex weight matrices for the STBC and should form a {\bf linearly independent} set over $\mathbb{R}$. It follows that
\begin{equation*}
\textbf{G} = [\widetilde{vec(\textbf{A}_1)}\ \widetilde{vec(\textbf{A}_2)}\ \cdots \ \widetilde{vec(\textbf{A}_{2k})}].
\end{equation*}

\begin{definition}\label{def5}({\bf Multigroup decodable STBCs}) An STBC is said to be $g$-group decodable \cite{sanjay} if its weight matrices can be separated into $g$ groups $\mathcal{G}_1$, $\mathcal{G}_2$, $\cdots$, $\mathcal{G}_g$ such that 
\begin{equation}\label{g_group}
\textbf{A}_i\textbf{A}_j^H + \textbf{A}_j\textbf{A}_i^H = \textbf{O}_{n_t}, ~~~~ \textbf{A}_i \in \mathcal{G}_l, \textbf{A}_j \in \mathcal{G}_p, l \neq p.
\end{equation}
\end{definition}

Equation \eqref{Y} can be rewritten as
\begin{equation*}
 \widetilde{vec(\textbf{Y})} = \sqrt{\frac{SNR}{n_t}}\textbf{H}_{eq}\textbf{s} + \widetilde{vec(\textbf{N})},
\end{equation*}
\noindent where $\textbf{H}_{eq} \in \mathbb{R}^{2n_rT\times 2n_{min}T}$ is given by
\begin{equation*}
 \textbf{H}_{eq} = \left(\textbf{I}_T \otimes \check{\textbf{H}}\right)\textbf{G},
\end{equation*}
with $\textbf{G} \in \mathbb{R}^{2n_tT\times 2n_{min}T}$ being the generator matrix as in Def. \ref{def4}. 

\section{Relationship between weight matrices and ergodic capacity}\label{sec3}

It has been shown that if the generator matrix is unitary, the STBC does not reduce the ergodic capacity of the MIMO channel \cite{damen_info}, \cite{JJK}. For the generator matrix to be unitary, a prerequisite is that the number of receive antennas should be atleast equal to the number of transmit antennas, because only then will the generator matrix be square. When $n_r < n_t$, only the Alamouti code has been known to achieve the ergodic capacity (by saying that an STBC achieves the ergodic capacity, we mean that with the use of a suitable outer code in conjuntion with the STBC, capacity can be achieved) of the $2 \times 1$ MIMO channel.  Since it is difficult to make an exact analysis of the ergodic capacity when $n_r < n_t$, we make an approximate analysis in the low and high SNR range. The ergodic capacity with the use of a space time code is given as follows \cite{JJK}.

\begin{equation}
 \mathcal{C} = \frac{1}{2T}\mathcal{E}_{\textbf{H}}logdet\left(\textbf{I}_{2n_rT} + \frac{SNR}{n_t}\textbf{H}_{eq}\textbf{H}_{eq}^T\right)
\end{equation} 

\subsection{Low SNR analysis}

Let $ \textbf{H}_{eq}\textbf{H}_{eq}^T = \textbf{U}\textbf{D}\textbf{U}^T$ be the singular value decomposition of $\textbf{H}_{eq}\textbf{H}_{eq}^T$. Let $\textbf{D} = diag[d_1,d_2,\cdots,d_{2Tn_r}]$ and $\textbf{H}_{eq} = [ \textbf{h}_1, \textbf{h}_2, \cdots, \textbf{h}_{2Tn_r}]$. We have,
\begin{eqnarray*}
 \mathcal{C}  & = & \frac{1}{2T}\mathcal{E}_{\textbf{H}}logdet\left(\textbf{I}_{2n_rT} + \frac{SNR}{n_t} \textbf{U}\textbf{D}\textbf{U}^T\right) \\
& = & \frac{1}{2T}\mathcal{E}_{\textbf{H}}logdet\left(\textbf{I}_{2n_rT} + \frac{SNR}{n_t} \textbf{D}\right) \\
& = & \frac{1}{2T}\mathcal{E}_{\textbf{H}}\left( log \prod_{i=1}^{2Tn_r}\left(1 + \frac{SNR}{n_t}d_i\right)\right) \\
& = & \frac{1}{2T}\mathcal{E}_{\textbf{H}}\left( \sum_{i=1}^{2Tn_r}log\left(1 + \frac{SNR}{n_t}d_i\right)\right) \\
& \approx & \frac{1}{2T}\mathcal{E}_{\textbf{H}} \left( \sum_{i=1}^{2Tn_r}\frac{SNR}{n_t}d_i\right) \\
& \approx & \frac{SNR}{2n_tT} \mathcal{E}_{\textbf{H}}\left( tr\left(\textbf{H}_{eq}\textbf{H}_{eq}^T\right)\right) \\
& \approx & \frac{SNR}{2n_tT} \mathcal{E}_{\textbf{H}}\left( \Vert \textbf{H}_{eq} \Vert^2\right) \\
& \approx & \frac{SNR}{2n_tT} \mathcal{E}_{\textbf{H}}\left( \sum_{i=1}^{2Tn_r} \Vert \textbf{h}_{i} \Vert^2\right).
\end{eqnarray*} 

Since $\textbf{h}_i = 
(\textbf{I}_T \otimes \check{\textbf{H}}) \widetilde{vec(\textbf{A}_i)}$, we have
\begin{equation}\label{erg}
\Vert \textbf{h}_{i} \Vert^2 = \Vert \textbf{HA}_{i} \Vert^2 = tr(\textbf{HA}_i\textbf{A}_i^H\textbf{H}^H).
\end{equation}

The ergodic capacity of an $n_t \times n_r$ MIMO channel is given as
\begin{equation}
 \mathcal{C}_{n_t \times n_r} = \mathcal{E}_{\textbf{H}}logdet(\textbf{I}_{n_r} + \frac{SNR}{n_t}\textbf{H}\textbf{H}^H).
\end{equation} 

In the low SNR scenario, 
\begin{equation}
 \mathcal{C}_{n_t \times n_r}  \approx \frac{SNR}{n_t} \mathcal{E}_{\textbf{H}} ( \Vert \textbf{H} \Vert^2 ).
\end{equation} 

Hence, in the low SNR scenario, if $\textbf{A}_i\textbf{A}_i^H = \frac{1}{n_r}\textbf{I}_{n_t}, \forall i = 1,2,\cdots,2Tn_r$, then, $\mathcal{C} = \mathcal{C}_{n_t \times n_r}$, which is evident from \eqref{erg}.

\subsection{High SNR analysis}
For this purpose, we use the $\textbf{QR}$ decomposition of $\textbf{H}_{eq}$. $\textbf{Q}$ and $\textbf{R}$ have the general form obtained by $Gram-Schmidt$ process as
\begin{equation*}
 \textbf{Q} \triangleq [ \textbf{q}_1\ \textbf{q}_2 \ \textbf{q}_3 \cdots \textbf{q}_{2Tn_r} ],
\end{equation*}
where $\textbf{q}_i, i = 1,2,\cdots,2Tn_r$ are column vectors, and
\begin{equation*}
 \textbf{R} \triangleq \left[\begin{array}{ccccc}
\Vert \textbf{r}_1 \Vert & \langle \textbf{q}_1,\textbf{h}_2 \rangle & \langle \textbf{q}_1,\textbf{h}_3 \rangle & \ldots &  \langle \textbf{q}_1,\textbf{h}_{2Tn_r} \rangle\\
0 & \Vert \textbf{r}_2 \Vert & \langle \textbf{q}_2,\textbf{h}_3 \rangle & \ldots & \langle \textbf{q}_2,\textbf{h}_{2Tn_r} \rangle\\
0 & 0 &  \Vert \textbf{r}_3 \Vert & \ldots & \langle \textbf{q}_3,\textbf{h}_{2Tn_r} \rangle\\
\vdots & \vdots & \vdots & \ddots & \vdots\\
0 & 0 & 0 & \ldots & \Vert \textbf{r}_{2Tn_r} \Vert\\
\end{array}\right]
\end{equation*}
\noindent where $\textbf{r}_1 = \textbf{h}_1$, $\textbf{q}_1 = \frac{\textbf{r}_1}{\Vert \textbf{r}_1 \Vert}$, $\textbf{r}_i = \textbf{h}_i - \sum_{j=1}^{i-1}\langle \textbf{q}_j,\textbf{h}_i \rangle \textbf{q}_j$ and $\ \textbf{q}_i = \frac{\textbf{r}_i}{\Vert \textbf{r}_i \Vert},\ i = 2,3,\cdots,2Tn_r$. We have,

\begin{eqnarray*}
 \mathcal{C}  & = & \frac{1}{2T}\mathcal{E}_{\textbf{H}}logdet\left(\textbf{I}_{2n_rT} + \frac{SNR}{n_t} \textbf{H}_{eq}\textbf{H}_{eq}^T\right) \\
& = & \frac{1}{2T}\mathcal{E}_{\textbf{H}}logdet\left(\textbf{I}_{2n_rT} + \frac{SNR}{n_t} \textbf{QR}\textbf{R}^T\textbf{Q}^T\right) \\
& = & \frac{1}{2T}\mathcal{E}_{\textbf{H}}logdet\left(\textbf{I}_{2n_rT} + \frac{SNR}{n_t} \textbf{R}\textbf{R}^T\right) \\
& \approx & \frac{1}{2T}\mathcal{E}_{\textbf{H}}logdet\left(\frac{SNR}{n_t} \textbf{R}\textbf{R}^T\right) \\
& \approx & n_r log\left(\frac{SNR}{n_t}\right)+\frac{1}{2T}\mathcal{E}_{\textbf{H}}logdet\left( \textbf{R}\textbf{R}^T\right).
\end{eqnarray*}

Using the well known identity that the determinant of a triangular matrix is the product of its diagonal elements, we have
\begin{eqnarray*}
 \mathcal{C}  & \approx & n_r log\left(\frac{SNR}{n_t}\right)+\frac{1}{2T}\mathcal{E}_{\textbf{H}}log \prod_{i=1}^{2Tn_r}\textbf{R}(i,i)^2 \\
& \approx & n_r log\left(\frac{SNR}{n_t}\right)+\frac{1}{2T}\mathcal{E}_{\textbf{H}}\sum_{i=1}^{2Tn_r}log \textbf{R}(i,i)^2. 
\end{eqnarray*}

From the definition of the $\textbf{R}$-matrix, we have, 
\begin{eqnarray*}
\textbf{R}(i,i)^2  & = & \Vert \textbf{r}_i \Vert^2\\
& = & \langle \textbf{r}_i,\textbf{r}_i\rangle \\
& = & \left\langle \left(\textbf{h}_i - \sum_{j=1}^{i-1}\langle \textbf{q}_j,\textbf{h}_i \rangle \textbf{q}_j\right), \left(\textbf{h}_i - \sum_{j=1}^{i-1}\langle \textbf{q}_j,\textbf{h}_i \rangle \textbf{q}_j\right)\right\rangle  \\
& = & \Vert \textbf{h}_i \Vert^2 - \sum_{j=1}^{i-1}\langle \textbf{q}_j,\textbf{h}_i \rangle^2.
\end{eqnarray*}

Hence,
\begin{equation}\label{highsnr}
 \mathcal{C}  \approx n_r log\left(\frac{SNR}{n_t}\right)+\frac{1}{2T}\sum_{i=1}^{2Tn_r}\mathcal{E}_{\textbf{H}}log \left(\Vert \textbf{h}_i \Vert^2 - \sum_{j=1}^{i-1}\langle \textbf{q}_j,\textbf{h}_i \rangle^2\right). 
\end{equation}

Equation \eqref{highsnr} tells us that at high $SNR$, the entries of the $\textbf{R}$-matrix, i.e, $\langle \textbf{q}_j,\textbf{h}_i \rangle$ dictate the ergodic capacity. If the number of zero entries in the upper block of the $\textbf{R}$-matrix is larger, then the ergodic capacity will be higher. Hence, it is essential that the $\textbf{R}$-matrix has as many zeros as possible. 

In \cite{SrR_arxiv} [Thm. 1], it has been shown that if $\textbf{A}_i\textbf{A}_j^H + \textbf{A}_j\textbf{
A}_i^H = \textbf{O}_{n_t}$, then, the $i^{th}$ and the $j^{th}$ columns of $\textbf{H}_{eq}$ are orthogonal. From the definition of $\textbf{R}$-matrix, column orthogonality of $\textbf{H}_{eq}$ dictates the presence of zeros. Hence, to design a good STBC when $n_r < n_t$, the equivalent channel matrix should have groups of columns orthogonal to one another. We would, of course, like all the columns to be orthogonal, but there is a limit to the number, the limit being the maximum number of Hurwitz-Radon matrices for $n_t$ transmit antennas. Except for the Alamouti code, this number is much lesser than $2Tn_r$, which is the number of weight matrices of a full-rate STBC when $n_r <n_t$. So, evidently, when $n_r < n_t$, higher ergodic capacity at high $SNR$ means lower ML-decoding complexity (because of column orthogonality). Hence, to construct an STBC with high ergodic capacity, we first construct rate-1 STBCs with the lowest possible ML-decoding complexity. So far, the known least ML-decoding complexity rate-1 codes are the rate-1, 4-group decodable codes. But the codes mentioned in literature \cite{4gp1}, \cite{4gp2}, \cite{sanjay} are not suitable for extension to higher number of receive antennas, since their design is obtained by iterative methods. Hence, in the next section, we propose a new design methodology to obtain the weight matrices of a rate-1, 4-group decodable code by algebraic methods for $2^a$ transmit antennas.

\section{Construction of Rate-1, 4-group decodable codes}\label{sec4}
We make use of the following Theorem, presented in \cite{4gp2}, to construct rate-1, 4-group decodable codes for $n = 2^a$ transmit antennas. 
\begin{theorem}
\label{thm1}
\cite{4gp2} An $n \times n$ linear dispersion code transmitting k real symbols is $g$-group decodable if the weight matrices satisfy the following conditions:
\begin{enumerate}
\item $\textbf{A}_i^2   =  \textbf{I}_n, i \in \{1,2,\cdots,\frac{k}{g}\}$.
\item $\textbf{A}_j^2   =  -\textbf{I}_n, j \in \{\frac{mk}{g} + 1, m = 1,2,\cdots,g-1\}$.
\item $\textbf{A}_i\textbf{A}_j = \textbf{A}_i\textbf{A}_j, i,j \in \{1,2,\cdots,\frac{k}{g} \}$.
\item $\textbf{A}_i\textbf{A}_j =  \textbf{A}_i\textbf{A}_j, i \in \{1,2,\cdots,\frac{k}{g} \}, j \in \{ \frac{mk}{g} + 1, m = 1,2,\cdots,g-1\}$.
\item $\textbf{A}_i\textbf{A}_j  =  -\textbf{A}_i\textbf{A}_j, i, j \in \{  \frac{mk}{g} + 1, m = 1,2,\cdots,g-1\}, i \neq j$.
\item $\textbf{A}_{\frac{mk}{g} + i} = \textbf{A}_i\textbf{A}_{\frac{mk}{g}+1}, m \in \{ 1,2,\cdots,g-1\}$, $i \in \{1,2,\cdots,g \} $.
\end{enumerate}

Table \ref{table1} illustrates the weight matrices of a $g$-group decodable code which satisfy the above conditions. The weight matrices in each column belong to the same group.
\begin{table*}
\begin{center}
\begin{tabular}{|l|l|l|l|} 
\hline
$\textbf{A}_1 = I_n$        &  $\textbf{A}_{\frac{k}{g}+1}$ & $\ldots$ &  $\textbf{A}_{\frac{(g-1)k}{g}+1}$ \\ \hline
$\textbf{A}_2$              &  $\textbf{A}_{\frac{k}{g}+2} = \textbf{A}_2\textbf{A}_{\frac{k}{g}+1}$ & $\ldots$ & $\textbf{A}_{\frac{(g-1)k}{g}+2} = \textbf{A}_2\textbf{A}_{\frac{(g-1)k}{g}+1}$ \\ \hline
$ \vdots$    & $\vdots$ & $\ldots$ & $\vdots$  \\ \hline
$\textbf{A}_{\frac{k}{g}}$  & $\textbf{A}_{\frac{2k}{g}} = \textbf{A}_{\frac{k}{g}}\textbf{A}_{\frac{k}{g}+1}$ & $\ldots$ & $\textbf{A}_{k} = \textbf{A}_{\frac{k}{g}}\textbf{A}_{\frac{(g-1)k}{g}+1}$   \\ \hline
\end{tabular}
\caption{Weight matrices of a $g$-group decodable code}
\label{table1}
\end{center}
\end{table*}

\end{theorem}

In order to obtain a Rate-1, 4-group decodable STBC for $2^a$ transmit antennas, it is sufficient if we have $2^{a+1}$ matrices satisfying the conditions in Theorem \ref{thm1}. To obtain these, we make use of the following lemmas. 

\begin{lemma}\label{lemma_1}
Consider $n \times n$ matrices with complex entries.
If $n=2^a$ and $n \times n$ matrices $\textbf{F}_i, i=1,2,\cdots,2a$ anticommute pairwise, then the set of products $\textbf{F}_{i_1}\textbf{F}_{i_2}\cdots \textbf{F}_{i_s}$ with $1 \leq i_1 < \cdots < i_s \leq 2a$ along with $\textbf{I}_n$ forms a basis for the $2^{2a}$ dimensional space of all $n \times n$ matrices over $\mathbb{C}$. 
\end{lemma}
\begin{proof}
Available in \cite{anti_matric}.
\end{proof}

\begin{lemma}\label{lemma_2}
If all the mutually anticommuting $n \times n$ matrices $\textbf{F}_i, i=1,2,\cdots,2a$ are unitary and anti-Hermitian, so that they square to $-\textbf{I}_n$, then the product $\textbf{F}_{i_1}\textbf{F}_{i_2}\cdots \textbf{F}_{i_s}$ with $1 \leq i_1 < \cdots < i_s \leq 2a$ squares to $(-1)^{\frac{s(s+1)}{2}}\textbf{I}_n$.
\end{lemma}
\begin{proof}
$(\textbf{F}_{i_1}\textbf{F}_{i_2} \cdots \textbf{F}_{i_s})(\textbf{F}_{i_1}\textbf{F}_{i_2} \cdots \textbf{F}_{i_s})$
\begin{eqnarray*}
 & = & (-1)^{s-1}(\textbf{F}_{i_1}^2 \textbf{F}_{i_2} \cdots \textbf{F}_{i_s})(\textbf{F}_{i_2}\textbf{F}_{i_3} \cdots \textbf{F}_{i_s}) \\
& = & (-1)^{s-1}(-1)^{s-2}(\textbf{F}_{i_1}^2 \textbf{F}_{i_2}^2 \cdots \textbf{F}_{i_s})(\textbf{F}_{i_3}\textbf{F}_{i_4} \cdots \textbf{F}_{i_s}) \\
& = & (-1)^{[(s-1)+(s-2)+\cdots 1]}(\textbf{F}_{i_1}^2 \textbf{F}_{i_2}^2 \cdots \textbf{F}_{i_s}^2) \\
& = & (-1)^{\frac{s(s-1)}{2}}(-1)^{s}\textbf{I}_n \\
& = & (-1)^{\frac{s(s+1)}{2}}\textbf{I}_n.
\end{eqnarray*}
Hence proved.
\end{proof}

\begin{lemma}\label{lemma_3}
Let $\textbf{F}_i, i = 1,2,\cdots,2a$ be anticommuting, anti-Hermitian, unitary matrices. Let $\Omega_1 = \{ \textbf{F}_{i_1},\textbf{F}_{i_2},\cdots,\textbf{F}_{i_s} \}$ and $\Omega_2 = \{ \textbf{F}_{j_1},\textbf{F}_{j_2},\cdots,\textbf{F}_{j_r} \}$
with $1 \leq i_1 < \cdots < i_s \leq 2a$ and $1 \leq j_1 < \cdots < j_r \leq 2a$. Let $\vert \Omega_1 \cap \Omega_2 \vert = p$. Then the product matrix $\textbf{F}_{i_1}\textbf{F}_{i_2}\cdots \textbf{F}_{i_s}$ commutes with $\textbf{F}_{j_1}\textbf{F}_{j_2}\cdots \textbf{F}_{j_r}$ if exactly one of the following is satisfied, and anticommutes otherwise.
\begin{enumerate}
\item $r,s$ and $p$ are all odd.
\item The product $rs$ is even and $p$ is even (including 0).
\end{enumerate}
\end{lemma}
\begin{proof}
For $\textbf{F}_{j_k} \in \Omega_1 \cap \Omega_2$, we note that
\begin{equation}
(\textbf{F}_{i_1}\textbf{F}_{i_2} \cdots \textbf{F}_{i_s})\textbf{F}_{j_k} = (-1)^{s-1}\textbf{F}_{j_k}(\textbf{F}_{i_1}\textbf{F}_{i_2} \cdots \textbf{F}_{i_s})
\end{equation}
\noindent and
\begin{equation}
(\textbf{F}_{i_1}\textbf{F}_{i_2} \cdots \textbf{F}_{i_s})\textbf{F}_{j_k} = (-1)^{s}\textbf{F}_{j_k}(\textbf{F}_{i_1}\textbf{F}_{i_2} \cdots \textbf{F}_{i_s})
\end{equation}
\noindent otherwise. Now,\\
$(\textbf{F}_{i_1}\textbf{F}_{i_2} \cdots \textbf{F}_{i_s})(\textbf{F}_{j_1}\textbf{F}_{j_2} \cdots \textbf{F}_{j_r})$
\begin{eqnarray*}
& = & (-1)^{p(s-1)}(-1)^{(r-p)s}(\textbf{F}_{j_1}\textbf{F}_{j_2} \cdots \textbf{F}_{j_r})(\textbf{F}_{i_1}\textbf{F}_{i_2} \cdots \textbf{F}_{i_s})\\
& = & (-1)^{rs-p}(\textbf{F}_{j_1}\textbf{F}_{j_2} \cdots \textbf{F}_{j_r})(\textbf{F}_{i_1}\textbf{F}_{i_2} \cdots \textbf{F}_{i_s}).
\end{eqnarray*}
\emph{case} 1). Since $r,s$ and $p$ are all odd, $(-1)^{rs-p}$ = 1.\\
\emph{case} 2). The product $rs$ is even and  $p$ is even (including 0). Hence $(-1)^{rs-p}$ = 1.
\end{proof}

\begin{lemma}\label{lemma4}
Let $\textbf{F}_i, i = 1,2,\cdots,2a$ be $2^a \times 2^a$ unitary, pairwise anticommuting matrices. Then, the product matrix $\textbf{F}_1^{\lambda_1}\textbf{F}_2^{\lambda_2}\cdots\textbf{F}_{2a}^{\lambda_{2a}}, \lambda_i \in \{0,1\}, i=1,2,\cdots, 2a$, with the exception of $\textbf{I}_{2^a}$, is traceless.
\end{lemma}
\begin{proof}
It is well known that $tr(\textbf{AB})=tr(\textbf{BA})$ for any two matrices $\textbf{A}$ and $\textbf{B}$. Let $\textbf{A}$ and $\textbf{B}$ be two invertible, $n \times n$ anticommuting matrices. So,
\begin{eqnarray*}
\textbf{A}\textbf{B} &=& -\textbf{B}\textbf{A}. \\ 
\textbf{A}\textbf{B}\textbf{A}^{-1} &=& -\textbf{B}.\\
tr(\textbf{A}\textbf{B}\textbf{A}^{-1}) &=& -tr(\textbf{B}).\\
tr(\textbf{A}^{-1}\textbf{A}\textbf{B}) = -tr(\textbf{B}) &\Leftrightarrow& tr(\textbf{B}) = -tr(\textbf{B}).
\end{eqnarray*}
\begin{equation}\label{trace}
\therefore tr(\textbf{B}) = 0.
\end{equation} 
Similarly, it can be shown that $tr(\textbf{A}) = 0$. By applying Lemma \ref{lemma_3}, it can be seen that any product matrix $\textbf{F}_1^{\lambda_1^\prime}\textbf{F}_2^{\lambda_2^\prime}\cdots \textbf{F}_{2a}^{\lambda_{2a}^\prime}$, anticommutes with some other product matrix from the set $\{\textbf{F}_1^{\lambda_1}\textbf{F}_2^{\lambda_2}\cdots\textbf{F}_{2a}^{\lambda_{2a}}, \lambda_i \in \{0,1\}, i = 1,2,3,\cdots,2a\}$. Hence, from the result obtained in \eqref{trace}, we can say that every product matrix $\textbf{F}_1^{\lambda_1}\textbf{F}_2^{\lambda_2}\cdots\textbf{F}_{2a}^{\lambda_{2a}}$ except $\textbf{I}_{2^a}$ is traceless.
\end{proof}

From Theorem \ref{thm1}, to get a rate-1, 4-group decodable STBC, we need 3 pairwise anticommuting, anti-Hermitian matrices which commute with a group of $2^{a-1}$ Hermitian, pairwise commuting matrices. Once these are identified, the other weight matrices can be easily obtained. From \cite{TiH}, one can obtain $2a$ pairwise anticommuting, anti-Hermitian matrices, presented here for completeness.

Let 
\begin{equation*}
\textbf{P}_1 =\left[\begin{array}{rr}
0 & 1 \\
-1 & 0 \\
\end{array}
\right],
\textbf{P}_2 =\left[\begin{array}{rr}
0 & j \\
j & 0 \\
\end{array}
\right], 
\textbf{P}_3 =\left[\begin{array}{rr}
1 & 0 \\
0 & -1 \\
\end{array}
\right]
\end{equation*}
and 
$\textbf{A}^{\otimes^{m}} \triangleq \underbrace{\textbf{A}\otimes \textbf{A}\otimes \textbf{A} \cdots \otimes \textbf{A} }_{m~~times  } $. \\
The $2a$ anti-Hermitian, pairwise anti-commuting matrices are 
\begin{eqnarray*}
\textbf{F}_1 &= &\pm j \textbf{P}_3^{\otimes^{a}}, \\
\textbf{F}_2 &= & \textbf{I}_2^{\otimes^{a-1}}  \bigotimes \textbf{P}_1, \\                                                                                
\textbf{F}_3 &= &\textbf{I}_2^{\otimes^{a-1}}  \bigotimes \textbf{P}_2, \\
. & . \\
. & . \\
\end{eqnarray*}
\begin{eqnarray*}
\textbf{F}_{2k} &= & \textbf{I}_2^{\otimes^{a-k}} \bigotimes  \textbf{P}_1 \bigotimes \textbf{P}_3^{\otimes^{k-1}}, \\
\textbf{F}_{2k+1} &= &\textbf{I}_2^{\otimes^{a-k}} \bigotimes \textbf{P}_2 \bigotimes \textbf{P}_3^{\otimes^{k-1}}, \\
. & . \\
. & . \\
\textbf{F}_{2a} &= &\textbf{P}_1 \bigotimes \textbf{P}_3^{\otimes^{a-1}}.          
\end{eqnarray*} 

For a set $\mathcal{S} = \{a_1,a_2,\cdots,a_n\}$, define $\mathbb{P}(\mathcal{S})$ as
\begin{equation*}
\mathbb{P}(\mathcal{S}) \triangleq \{a_1^{\lambda_1}a_2^{\lambda_2}\cdots a_n^{\lambda_n}, \lambda_i \in \{0,1\}\}.
\end{equation*} 

We choose $\textbf{F}_1$, $\textbf{F}_2$ and $\textbf{F}_3$ to be the three pairwise anticommuting, anti-Hermitian matrices (to be placed in the top row along with $\textbf{I}_n$ in Table \ref{table1}. Consider the set $\mathcal{S} = \{j\textbf{F}_4\textbf{F}_5, j\textbf{F}_6\textbf{F}_7,$ $\cdots, j\textbf{F}_{2a-2}\textbf{F}_{2a-1}, \textbf{F}_1\textbf{F}_2\textbf{F}_3 \}$, the cardinality of which is $a-1$. Using Lemma \ref{lemma_2} and Lemma \ref{lemma_3}, one can note that the set consists of pairwise commuting matrices which are Hermitian. Moreover, one can note that each of the matrices in the set also commutes with $\textbf{F}_1$, $\textbf{F}_2$ and $\textbf{F}_3$. Hence,  $\mathbb{P}(\mathcal{S})$, which has cardinality $2^{a-1}$ is also a set with pairwise commuting, Hermitian matrices which also commute with $\textbf{F}_1$, $\textbf{F}_2$ and $\textbf{F}_3$. The linear independence of $\mathbb{P}(\mathcal{S})$ over $\mathbb{C}$ is easy to see by applying Lemma \ref{lemma_1}. Hence,  we have 3 pairwise anticommuting, anti-Hermitian matrices which commute with a group of $2^{a-1}$ Hermitian, pairwise commuting matrices. Having obtained these, the other weight matrices are obtained from Theorem \ref{thm1}. 

\subsection{An example - $n=8$}
 To illustrate with an example, we consider the case $n = 8$. Let $\textbf{F}_i, i=1,2,\cdots,6$ denote the 6 pairwise anticommuting, anti-Hermitian matrices. Choose $\textbf{F}_1$, $\textbf{F}_2$ and $\textbf{F}_3$ to be the three anticommuting matrices. Let
\begin{equation*}
\mathcal{S} = \{j\textbf{F}_4\textbf{F}_5, \textbf{F}_1\textbf{F}_2\textbf{F}_3\}
\end{equation*}
and 
\begin{equation*}
\mathbb{P}(\mathcal{S}) = \{\textbf{I}_8, j\textbf{F}_4\textbf{F}_5, \textbf{F}_1\textbf{F}_2\textbf{F}_3, j\textbf{F}_1\textbf{F}_2\textbf{F}_3\textbf{F}_4\textbf{F}_5\}.
\end{equation*}

The 16 weight matrices of the rate-1, 4-group decodable code for 8 antennas are 
as shown below. Each column corresponds to the weight matrices in a group. Note that the product of any two matrices in the first group is some other matrix in the same group.

\begin{center}
\begin{tabular}{|c|c|c|c|} 
\hline
$\textbf{I}_8$   &  $\textbf{F}_1$   &  $\textbf{F}_2$  & $\textbf{F}_3 $ \\ \hline
$j\textbf{F}_4\textbf{F}_5$   & $j\textbf{F}_1\textbf{F}_4\textbf{F}_5$  & $j\textbf{F}_2\textbf{F}_4\textbf{F}_5$ &  $j\textbf{F}_3\textbf{F}_4\textbf{F}_5$ \\ \hline
$\textbf{F}_1\textbf{F}_2\textbf{F}_3$ & $-\textbf{F}_2\textbf{F}_3$  & $\textbf{F}_1\textbf{F}_3$  & $-\textbf{F}_1\textbf{F}_2$ \\ \hline
$j\textbf{F}_1\textbf{F}_2\textbf{F}_3\textbf{F}_4\textbf{F}_5$   &  $-j\textbf{F}_2\textbf{F}_3\textbf{F}_4\textbf{F}_5$ & $j\textbf{F}_1\textbf{F}_3\textbf{F}_4\textbf{F}_5$  & $-j\textbf{F}_1\textbf{F}_2\textbf{F}_4\textbf{F}_5$ \\ \hline
\end{tabular}
\end{center}

\subsection{Coding gain calculations}\label{sec4A}
Let $\Delta(\textbf{S},\textbf{S}^\prime) \triangleq det\big(\Delta \textbf{S}\Delta \textbf{S}^H\big)$, where $\Delta \textbf{S} \triangleq \textbf{S} - \textbf{S}^\prime, \textbf{S} \neq \textbf{S}^\prime $ denotes the codeword difference matrix. Let $\Delta s_i \triangleq s_i - s_i^{\prime}, i =1,2,\cdots,2n_t$, where $s_i$ and $s_i^{\prime}$ are the real symbols encoding codeword matrices $\textbf{S}$ and $\textbf{S}^\prime$, respectively. Hence,

\begin{eqnarray*}
\Delta(\textbf{S},\textbf{S}^\prime) & = & det\left(\sum_{i=1}^{2n_t}\Delta s_i \textbf{A}_i \sum_{m=1}^{2n_t}\Delta s_m \textbf{A}_m^H\right) \\
& = & det\left(\sum_{i=1}^{2n_t} \sum_{m=1}^{2n_t} \Delta s_i \Delta s_m \textbf{A}_i \textbf{A}_m^H\right). 
\end{eqnarray*}
Note that because of the nature of construction of the weight matrices,
\begin{equation*}
\textbf{A}_i\textbf{A}_m^H = \textbf{A}_{\frac{pn_t}{2}+i}\textbf{A}_{\frac{pn_t}{2}+m}^H, i,m \in \{1,2,3,4\}~~ p \in \{1,2,3\}.
\end{equation*}
Further, since the code is 4-group decodable,
\begin{eqnarray*}
\Delta(\textbf{S},\textbf{S}^\prime) & = & det\left(\sum_{ p=0}^{3}\left(\sum_{i=\frac{pn_t}{2}+1}^{\frac{(p+1)n_t}{2}} \Delta s_i^2\textbf{I}_{n_t} +2 \sum_{i=\frac{pn_t}{2}+1}^{\frac{(p+1)n_t}{2}-1}\sum_{m=i + 1}^{\frac{(p+1)n_t}{2}}\Delta s_i\Delta s_m \textbf{A}_i\textbf{A}_m^H\right)\right).
\end{eqnarray*}
Since all the weight matrices in the first group are Hermitian and pairwise commuting and the product of any two such matrices is some other matrix in the same group.  It is well known that commuting matrices are simultaneously diagonalizable. Hence,
\begin{equation*}
\textbf{A}_i = \textbf{E}\textbf{D}_i\textbf{E}^H, i \in \left\{2,3,\cdots, \frac{n_t}{2} \right\},
\end{equation*} 
where, $\textbf{D}_i$ is a diagonal matrix. Since $\textbf{A}_i$ is Hermitian as well as unitary, the diagonal elements of $\textbf{D}_i$ are $\pm 1$. In addition, from Lemma \ref{lemma4}, there is an equal number of '1's and '-1's.
In fact, because of the nature of construction of the matrices $\textbf{F}_i,i=1,2,\cdots,2a$, the product matrices $\textbf{F}_i\textbf{F}_{i+1}$, for even $i$, and the product matrix  $\textbf{F}_1\textbf{F}_{2}\textbf{F}_{3}$ are always diagonal (easily seen from the definition of $\textbf{F}_i$, $i=1,2,\cdots,2a$). Hence, all the weight matrices of the first group excluding $\textbf{A}_1 = \textbf{I}_{n_t}$ are diagonal with the diagonal elements being $\pm 1 $. Since these diagonal matrices also commute with $\textbf{F}_2$ and $\textbf{F}_3$, the diagonal entries are such that for every odd $i$, if the $(i,i)^{th}$ entry is 1(-1), then, the $(i+1,i+1)^{th}$ entry is also 1(-1, resp.).  To summarize, the properties of $\textbf{A}_i$, $i=2,\cdots, \frac{n_t}{2}$ are listed below.
\begin{eqnarray}
\label{p1}
 \textbf{A}_i & = &  \textbf{A}_i^H \\
 \label{p2} 
\textbf{A}_i^2 & = & \textbf{I}_{n_t} \\ 
\label{p3}
\textbf{A}_i(m,n) & = & 0, m \neq n \\ 
\label{p4} 
\textbf{A}_i(j,j) & = & \pm 1,j=1,2,\cdots, n_t \\
\label{p5}
tr(\textbf{A}_i) & = & 0 \\
\label{p6}
\textbf{A}_i(j,j) & = & \textbf{A}_i(j+1,j+1), j = 1, 3, 5,\cdots,n_t-1 \\
\label{p7}
\textbf{A}_i \textbf{A}_j & = & \textbf{A}_k, ~~~~ i,j,k \in \left\{1,2,\cdots,\frac{n_t}{2}\right\}.
\end{eqnarray}

In view of these properties,

\begin{equation*}
\Delta(\textbf{S},\textbf{S}^\prime) =  det\left(\sum_{p=0}^{3}\left(\sum_{i=\frac{pn_t}{2}+1}^{\frac{(p+1)n_t}{2}} \Delta s_i^2\textbf{I}_{n_t} + 2\sum_{i=\frac{pn_t}{2}+1}^{\frac{(p+1)n_t}{2}-1}\sum_{m=i + 1}^{\frac{(p+1)n_t}{2}}\Delta s_i \Delta s_m \textbf{D}_{im}\right) \right),
\end{equation*}
where, $\textbf{D}_{im} = \textbf{A}_k$ for some $k \in \{1,2,\cdots,\frac{n_t}{2} \}$. So,

\begin{equation*}
\Delta(\textbf{S},\textbf{S}^\prime) =  det\big( \prod_{j=1}^{n_t}\sum_{p=0}^{3}( \sum_{i=1}^{\frac{n_t}{2}}d_{ij} \Delta s_{\frac{pn_t}{2}+i})^2\big) 
\end{equation*}

where, $d_{ij} = \pm 1$ and $d_{1j} = 1$. In fact, $d_{ij} = \textbf{A}_i(j,j)$, $i = 1,2, 3, \cdots, \frac{n_t}{2}$. 
Hence,
\begin{equation*}
\min_{\textbf{S},\textbf{S}^\prime}(\Delta(\textbf{S},\textbf{S}^\prime))  =  det\left(\prod_{j=1}^{n_t}\left(\sum_{i=1}^{\frac{n_t}{2}}d_{ij} \Delta s_i\right)^2\right).  
\end{equation*}

From \eqref{p6},
\begin{equation}\label{tem}
\min_{\textbf{S},\textbf{S}^\prime}(\Delta(\textbf{S},\textbf{S}^\prime))  =  det\big(\prod_{j=1}^{\frac{n_t}{2}}(\sum_{i=1}^{\frac{n_t}{2}}d_{i(2j-1)} \Delta s_i)^4\big).  
\end{equation}

We need the minimum determinant to be as high a non-zero number as possible. In this regard, let 
\begin{equation}\label{wmatrix}
\textbf{W} \triangleq \sqrt{\frac{2}{n_t}}[w_{ij}], w_{ij} = d_{i(2j-1)}, i,j = 1,2,\cdots, \frac{n_t}{2} 
\end{equation}
and 
\begin{equation*}
\textbf{y}_p \triangleq [y_{\frac{n_tp}{2}+1},y_{\frac{n_tp}{2}+2},\cdots, y_{\frac{n_t(p+1)}{2}}]^T = \textbf{W} [s_{\frac{n_tp}{2}+1},s_{\frac{n_tp}{2}+2},\cdots, s_{\frac{n_t(p+1)}{2}}]^T, p = 0,1,2,3.
\end{equation*}

\begin{lemma}
$\textbf{W}$ as defined in \eqref{wmatrix} is a unitary matrix.
\end{lemma}
\begin{proof}
From \eqref{wmatrix}, it can be noted that the columns of $\textbf{W}$ are obtained from the diagonal elements $\textbf{A}_i$, $i=1,2,\cdots,\frac{n_t}{2}$. Each element of a column $i$ of $\textbf{W}$ corresponds to every odd numbered diagonal element of $\textbf{A}_i$. Denote the $i^{th}$ column of $\textbf{W}$ by $\textbf{w}_i$. Applying \eqref{p6}, \eqref{p7} and \eqref{p5} in that order,
\begin{eqnarray*}
\langle \textbf{w}_i,\textbf{w}_j \rangle & = & tr(\textbf{A}_i\textbf{A}_j) \\
& = & tr(\textbf{A}_k) \\
 & = & \delta_{ij}
\end{eqnarray*} 
where 
\begin{equation*}
\delta_{ij} = \left\{ \begin{array}{ccc}
 0 & \textrm{if} & i \neq j\\
 1 & \textrm{otherwise} & \\
\end{array} \right. 
\end{equation*} 
Hence, $\textbf{W}$ is unitary.
\end{proof}

Substituting $\textbf{y}_p$ in \eqref{tem}, we get 
\begin{equation*}
\min_{\textbf{S},\textbf{S}^\prime}(\Delta(\textbf{S},\textbf{S}^\prime))  =  det\left(\prod_{j=1}^{\frac{n_t}{2}} y_j^4\right). 
\end{equation*}

So, the minimum determinant is a power of the minimum product distance in $n_t/2 $ real dimensions. If $\textbf{y}_p \in \mathbb{Z}^{\frac{n_t}{2}}$, the product distance can be maximized by premultiplying $\textbf{y}_p$ with a suitable unitary rotation matrix $\textbf{U}$ given in \cite{full_div_rot}. This operation maximizes the minimum determinant and hence the coding gain. So, the real symbols of the rate-1, 4-group decodable code are encoded by grouping $\frac{n_t}{2}$ real symbols into 4 groups and each group of symbols taking value from a unitarily rotated vector belonging to $\mathbb{Z}^{\frac{n_t}{2}}$, the rotation matrix being $\textbf{W}^H\textbf{U}$. The ML-decoding complexity of the code is $M^{\frac{n_t}{4}}$, where $M$ is the size of the complex signal constellation. This is because there are $n_t/2$ real symbols per group to be jointly decoded and if we assume a complex constellation of size $M$, the number of complex symbols to be jointly decoded is $n_t/4$.

\section{Extension to higher number of receive antennas}\label{sec5}
When $n_r = 1$, a rate-1, 4-group decodable STBC is the best full-rate STBC possible in terms of ML-decoding complexity and as a result, ergodic capacity. However, when $n_r > 1$, we need more weight matrices to meet the full-rate criterion. Let $n_t = 2^a$. We know that if $\textbf{F}_i,i=1,2,\cdots,2a$ are pairwise anticommuting, invertible matrices, then, the set $\mathcal{F} \triangleq  \{\textbf{F}_1^{\lambda_1}\textbf{F}_2^{\lambda_2}\cdots\textbf{F}_{2a}^{\lambda_{2a}}$, with $\lambda_i \in \{0,1\}, i = 1,2,\cdots,2a\}$ is a linearly independent set over $\mathbb{C}$. Hence, the set $\mathcal{M} = \{ \mathcal{F}, j\mathcal{F} \}$ is linearly independent over $\mathbb{R}$. As a result, the elements of $\mathcal{M}$ can be used as weight matrices of a full-rate STBC for $n_r > 1$. Keeping in view that the ergodic capacity depends on as many non-diagonal entries of the $\textbf{R}$-matrix being zeros, it is important to choose the weight matrices prudently. The idea is that given a full-rate STBC for $n_r-1$ receive antennas, obtain the additional weight matrices of a full-rate STBC for $n_r$ receive antennas by using the weight matrices of a rate-1, 4-group decodable STBC such that the after addition of the new weight matrices, the set of weight matrices is linearly independent over $\mathbb{R}$. This is achieved as follows.
\begin{enumerate}
\item Obtain a rate-1, 4-group decodable STBC by using the construction method detailed in Section \ref{sec4}. Due to the nature of construction, the product of any two weight matrices is always some other weight matrix of the code, up to negation. Denote the set of weight matrices by $\mathcal{G}_1$. The ML-decoding complexity of the code is $M^{\frac{n_t}{4}}$.
\item From the set $\mathcal{F}$, choose a matrix that does not belong to $\mathcal{G}_1$ and multiply it with the elements of $\mathcal{G}_1$ to obtain a new set of weight matrices, denoted by $\mathcal{G}_2$. Clearly, the two sets will not have any matrix in common. The weight matrices of $\mathcal{G}_2$ form a new, rate-1, 4-group decodable STBC. This is because the ML-decoding complexity does not change by multiplying the weight matrices of a code with a unitary weight matrix. In this case, we have multiplied the elements of $\mathcal{G}_1$ with an element of $\mathcal{F}$, which is a unitary matrix. Now, $\mathcal{G}_1 \bigcup \mathcal{G}_2$ is the set of weight matrices of a rate-2 code with an ML-decoding complexity of $M^{n_t}.M^{\frac{n_t}{4}} = M^{\frac{5n_t}{4}}$. This is achieved by decoding the last $n_t$ symbols with a complexity of $M^{n_t}$ and then conditionally decoding the first $n_t$ symbols using the 4-group decodability property.
\item For increasing number of $n_r$, repeat as in the second step, obtaining new rate-1, 4-group decodable codes and then appending their weight matrices to obtain a new, rate-$n_r$ code with an ML-decoding complexity of  $M^{n_t(n_r-\frac{3}{4})}$.
\item When all the elements of $\mathcal{F}$ have been exhausted (this occurs when $n_r = n_t/2$), the remaining matrices up to a rate of $n_t$ symbols per channel use can be obtained from $j\mathcal{F}$. Note from Lemma \ref{lemma_1} that this does not spoil the linear independence over $\mathbb{R}$ of the weight matrices.
\end{enumerate}

The $\textbf{R}$-matrix of the STBC for $n_r$ receive antennas has the following structure, irrespective of the channel realization.
\begin{equation}\label{rmatrix}
 \textbf{R} =  \left[\begin{array}{cccc}
\textbf{D} &  \textbf{X} &  \ldots &  \textbf{X} \\
\textbf{O}_{2n_t} &  \textbf{D} &  \ldots &  \textbf{X} \\
\vdots &   \ddots  & \ddots & \vdots \\
\textbf{O}_{2n_t} & \textbf{O}_{2n_t} & \ldots & \textbf{D} \\
\end{array} \right]
\end{equation}
where $\textbf{X} \in \mathbb{R}^{2n_t \times 2n_t}$ is a random non-sparse matrix whose entries depend on the channel coefficients and $\textbf{D}  = \textbf{I}_4 \otimes \textbf{V} $, with $\textbf{V} \in \mathbb{R}^{\frac{n_t}{2} \times \frac{n_t}{2}}$ being an upper triangular matrix. As a result of the structure of $\textbf{D}$, the $\textbf{R}$-matrix has a large number of zeros in the upper block, and hence, compared to other existing codes, the proposed codes are expected to have higher ergodic capacity (for $n_r < n_t$) and lower average ML-decoding complexity.

\subsection{The Silver code as a special case of $n_t =2$}

The silver code, which is well known for being a low complexity, full-rate, full-diversity STBC for $n_r \geq 2$, transmits 2 complex symbols per channel use. Its first four weight matrices are that of the Alamouti code, which is a rate-1, 4-group decodable STBC for 2 transmit antennas. The Silver code's next 4 weight matrices are obtained by multiplying the first four weight matrices with $j$. However, to make the code a full-ranked one, the last four symbols take values from a different constellation, which is obtained by unitarily rotating the symbol vector in $\mathbb{Z}[j]^2$. The Silver code compares very well with the well known Golden code in error performance, while offering lower ML-decoding complexity of $M^2$ for square-QAM only.

\section{Simulation results}\label{sec6}
In all the simulation scenarios in this section, we consider the Rayleigh block fading MIMO channel. We consider 8 transmit antennas. To construct a rate-2 code for 8 transmit antennas, we first construct a rate-1, 4-group decodable STBC as described in Section \ref{sec4} and denote the set of obtained weight matrices by $\mathcal{G}_1$. Next we multiply the weight matrices of $\mathcal{G}_1$ by $\textbf{F}_4$ to obtain a new set of weight matrices which is denoted by $\mathcal{G}_2$. The weight matrices of the new STBC with rate-2 are obtained from $\mathcal{G}_1 \bigcup \mathcal{G}_2$. A rate-3 code for 3 receive antennas can be obtained by multiplying the matrices of $\mathcal{G}_1$ with $\textbf{F}_6$ and appending the resulting weight matrices to the set $\mathcal{G}_1 \bigcup \mathcal{G}_2$. The rival code is the punctured perfect code for 8 transmit antennas \cite{new_per}. The ergodic capacity plots of the two codes are shown in Fig. \ref{fig_cap}. As expected, our code achieves higher ergodic capacity, although lower than that of the corresponding MIMO channel. It must however, be noted that both codes help to achieve the same ergodic capacity as that of the MIMO channel for $n_r \geq 8$ because the generator matrix is unitary in that case.  

Fig. \ref{fig1} shows the codeword error performance of our code for $8 \times 2$ system and the punctured perfect code using 4-QAM. The performance is more or less the same. Unlike the Perfect code, our code will not have full-diversity if the design is made as explained in Section \ref{sec4}. However, we have multiplied the weight matrices of $\mathcal{G}_2$ with the scalar $e^{\frac{j\pi}{4}}$ in order to enhance performance. Fig. \ref{fig2} shows the symbol error performances of the two codes. Our code has a better performance at a higher SNR. The reason for this is that the number of symbol errors per codeword error for the punctured Perfect code is more than that of our code. Hence, even though the CER is the same, the SER is different. Our code appears to have full diversity, but we have not been able to prove it. The most important aspect of our code is that it has an ML-decoding complexity of $M^{10}$, while that of the comparable punctured Perfect code is $M^{16}$.

\begin{figure}
\centering
\includegraphics[width=4.5in,height=4in]{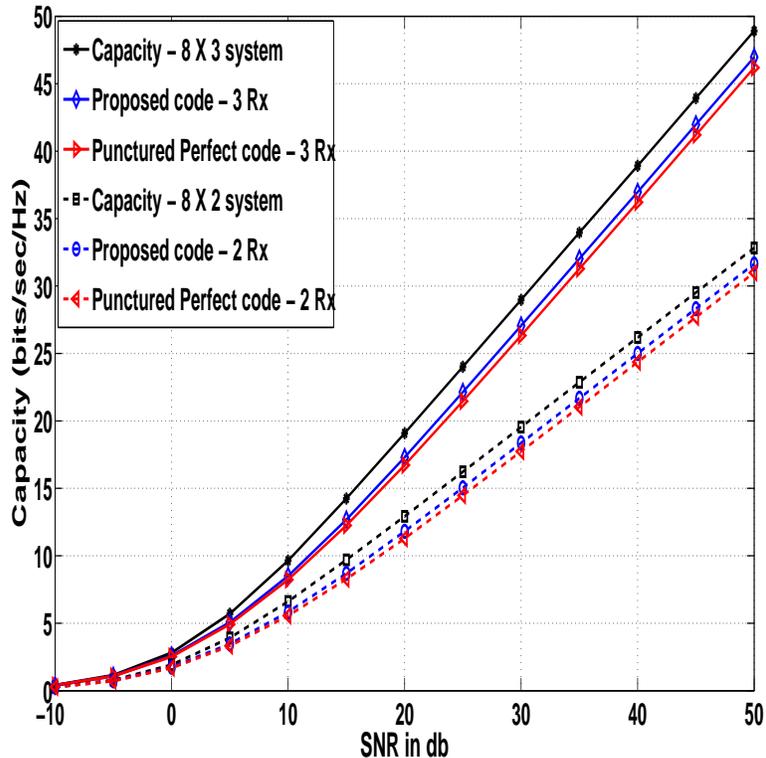}
\caption{Ergodic capacity Vs SNR for codes for $8 \times 2$ and $8 \times 3$ systems}
\label{fig_cap}
\end{figure}

\begin{figure}
\centering
\includegraphics[width=4.5in,height=4in]{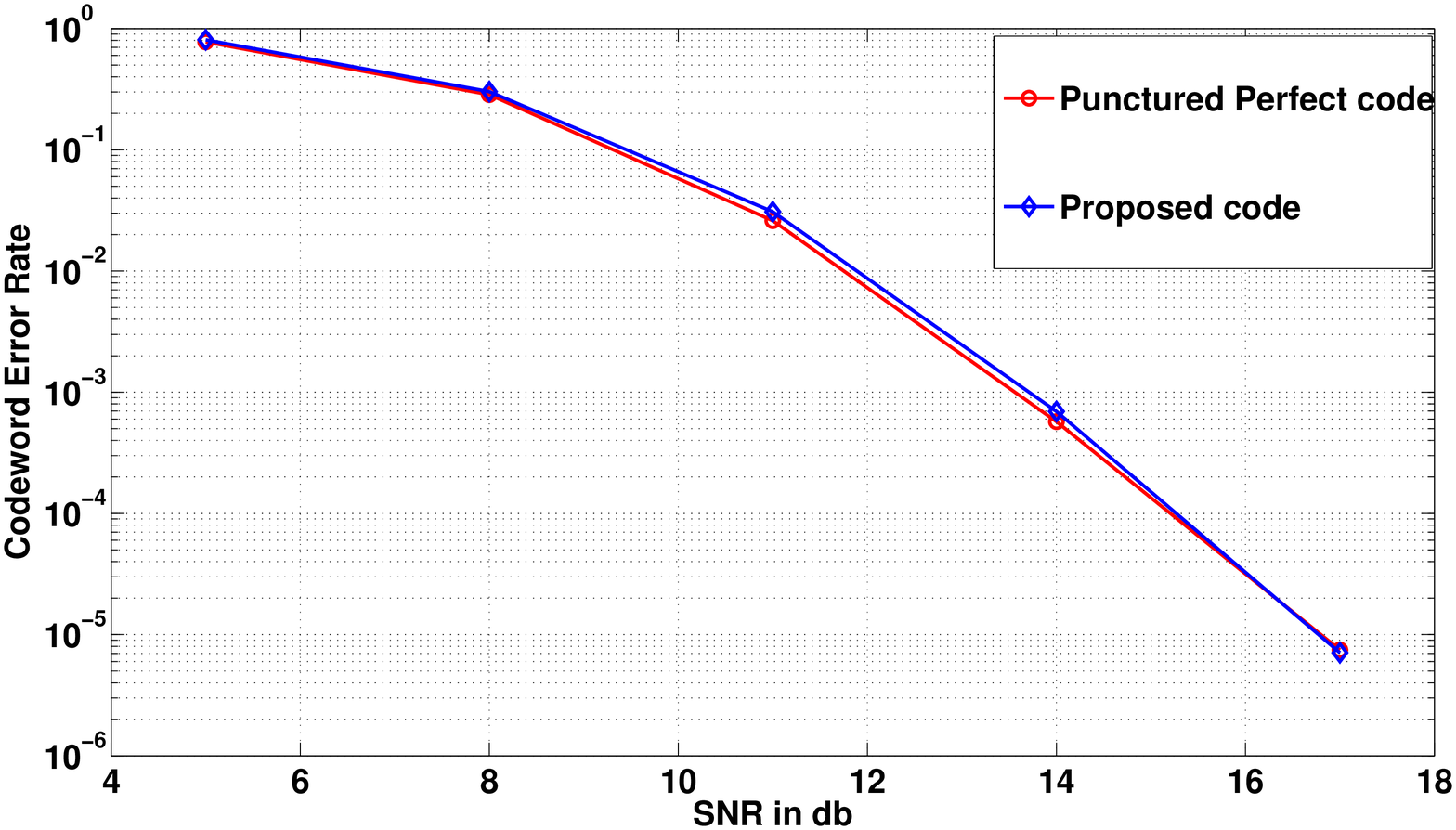}
\caption{CER performance at 4 BPCU for codes for $8 \times 2$ systems}
\label{fig1}
\end{figure}

\begin{figure}
\centering
\includegraphics[width=4.5in,height=4in]{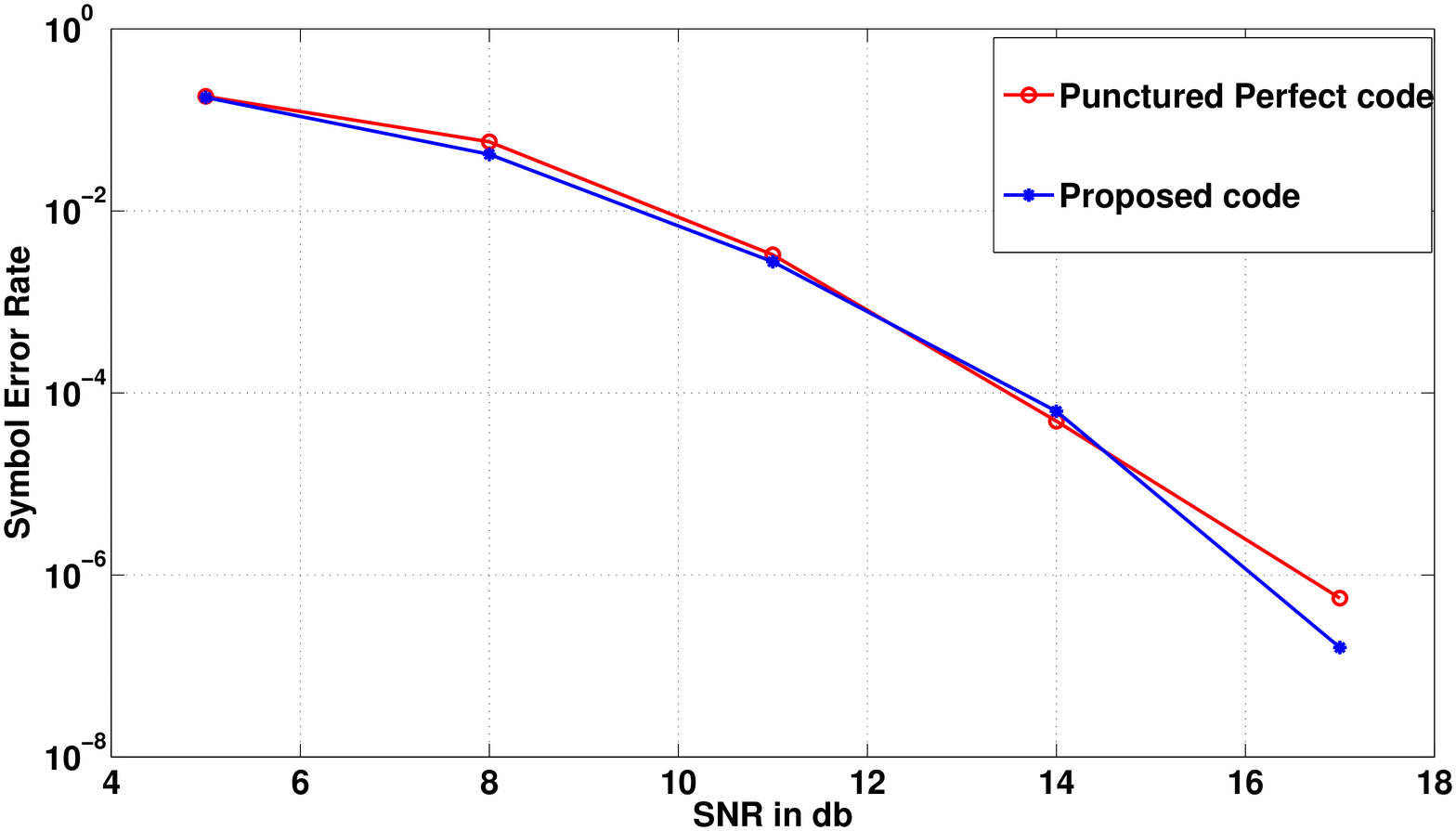}
\caption{SER performance at 4 BPCU for codes for $8 \times 2$ systems}
\label{fig2}
\end{figure}

\section{Discussion} \label{sec7}

In this paper, we proposed a scheme to obtain a full-rate STBC for $2^a$ transmit antennas and any number of receive antennas with reduced ML-decoding complexity. The STBCs thus obtained have higher ergodic capacity at high $SNR$ than existing STBCs for the case $n_r < n_t$. We have, however, not been able to provide a scheme to obtain full-diversity codes from these designs. Also it is to be seen if the proposed codes are better suited than existing codes for sub-optimal decoding techniques like lattice reduction aided detection, owing to the fact that more number of symbols are disentangled from one another than in the case of known codes. These are some of the directions for future research. 

\section*{ACKNOWLEDGEMENT}
This work was partly supported by the DRDO-IISc program on Advanced Research in Mathematical Engineering through research grants and the INAE Chair Professorship to B. Sundar Rajan.

\end{document}